\documentclass{article}
\usepackage{amsfonts}

\newtheorem{theorem}{Theorem}

\newtheorem{corollary}[theorem]{Corollary}

\newtheorem{lemma}[theorem]{Lemma}

\newtheorem{proposition}[theorem]{Proposition}

\newenvironment{proof}[1][Proof]{\noindent\textbf{#1.} }{\ \rule{0.5em}{0.5em}}
\input{tcilatex}
\begin{document}

\title{On the condition of conversion of classical probability distribution
families into quantum families}
\author{Keiji Matssumoto \\
National Institute of Informatics, \\
2-1-2, Hitotsubashi, Chiyoda-ku, Tokyo 101-8430\\
keiji@nii.ac.jp}
\maketitle

\begin{abstract}
The purpose of the paper is to study the condition for a probability
distribution family to a quantum state family. This is an (relatively) easy
example of quantum version of "comparison of statistical experiments", which
had turned out to supply deep insight into the foundation of classical and
quantum statistics \cite{Strasser}\cite{Torgersen}. It turns out use of
maximal quantum $f$-divergence is useful in characterizing the
classical-quantum transformability.
\end{abstract}

\section{Introduction}

The purpose of the paper is to study the condition for a probability
distribution family to a quantum state family. This is an (relatively) easy
example of quantum version of "comparison of statistical experiments", which
had turned out to supply deep insight into the foundation of classical and
quantum statistics \cite{Strasser}\cite{Torgersen}. Consideration of such a
problem nicely characterizes some known statistically important quantities,
giving their new operational meaning and proving some of their properties in
smart way. For example, RLD Fisher information, which is known to be the
achievable lower bound to the minimum mean square error of the Gaussian
shift model\thinspace \cite{Holevo:82} and `coherent' pure state
models\thinspace \cite{FujiwaraNagaoka:96}\cite{Matsumoto:pure}, is
characterized as the smallest classical Fisher information to simulate the
quantum statistical model locally. Also, a version of quantum relative
entropy, first studied by \cite{Belavkin}, turned out to be the smallest
classical entropy to generate two point quantum state family\cite%
{Matsumoto:relative}. Due to this characterization, this version of quantum
relative entropy had turned out to be the largest monotone relative entropy
which coincide with its classical counter part in commutative case\cite%
{Matsumoto:relative}.

Below, we describe our setting precisely. In the paper, the dimension of
Hilbert space $\mathcal{H}$ is finite. $\mathcal{L}\left( \mathcal{H}\right) 
$ is the set of all linear transforms on the Hilbert space $\mathcal{H}$. A
trace preserving completely positive map from a finite dimensional
commutating matrices (which is interpreted as a real function over a finite
set) to $\mathcal{L}\left( \mathcal{H}\right) $ is called
classical-to-quantum (CQ) map. We consider parameterized family of
probability distributions and quantum states, where the parameter space is
binary set, $\Theta :=\left\{ 0,1\right\} $. Hereafter, a probability
density function $p$ over the finite set $\mathcal{X}$ is always identified
with the finite dimensional matrix 
\[
\sum_{x\in \mathcal{X}}p\left( x\right) \left\vert e_{x}\right\rangle
\left\langle e_{x}\right\vert , 
\]%
where $\left\{ \left\vert e_{x}\right\rangle \right\} _{x\in \mathcal{X}}$
is an orthonormal set of vectors. Our problem is to investigate the
conditions for the existence of a CPTP map $\Gamma $ with

\begin{equation}
\Gamma \left( p_{\theta }\right) =\sigma _{\theta },\,\forall \theta \in
\Theta ,  \label{map}
\end{equation}%
where $\left\{ p_{\theta }\right\} _{\theta \in \Theta }$ and $\left\{
\sigma _{\theta }\right\} _{\theta \in \Theta }$ are given family of
probability density functions and density operators, respectively.

\section{$f$-Divergence}

When $\left\{ \sigma _{\theta }\right\} _{\theta \in \Theta }$ is also
commutative, i.e., the condition for classical-to-classical conversion is
well-studied, and the necessary and sufficient condition is characterized by 
$f$-divergence; Given a convex function $f$ on $[0,\infty )$, $f$-divergence
between probability distributions $p_{0}$ and $p_{1}$ is 
\[
D_{f}\left( p_{0}||p_{1}\right) :=\sum_{x\in \mathrm{supp\,}%
p_{1}}\,p_{1}\left( x\right) f\left( \frac{p_{0}\left( x\right) }{%
p_{1}\left( x\right) }\right) +\left( \sum_{x\not\in \mathrm{supp\,}%
p_{1}}\,p_{0}\left( x\right) \right) \lim_{\lambda \rightarrow \infty }\frac{%
f\left( \lambda \right) }{\lambda }. 
\]

\begin{lemma}
\label{lem:c-c-divergence}((\cite{Torgersen-finite}),(\cite{Torgersen}%
))There is a transition probability matrix $P$ with 
\[
Pp_{\theta }=q_{\theta }\,,\,\forall \theta \in \Theta 
\]%
exists if and only if 
\[
D_{f}\left( p_{0}||p_{1}\right) \geq D_{f}\left( q_{0}||q_{1}\right) 
\]%
holds for any proper and closed convex function $f$ on $[0,\infty )$.
\end{lemma}

Motivated by the above Lemma, we study the relation between the condition (%
\ref{map}) and a quantum version of $f$-divergence. Among many quantum
versions of $f$-divergence, we use the following one which is defined using
classical-quantum conversion problem: 
\[
D_{f}^{\max }\left( \sigma _{0}||\sigma _{1}\right) :=\min \{D_{f}\left(
p_{0}||p_{1}\right) ;\left\{ p_{\theta }\right\} _{\theta \in \Theta
},\,\Gamma \text{: CPTP with (\ref{map})}\}. 
\]

This quantity can be written more or less explicitly, if $f$ is an operator
convex function on $[0,\infty )$: 
\begin{equation}
D_{f}^{\max }\left( \sigma _{0}||\sigma _{1}\right) =\mathrm{tr}\,\sigma
_{1}f\left( \sigma _{1}^{-1/2}\tilde{\sigma}_{0}\sigma _{1}^{-1/2}\right)
+\left( 1-\mathrm{tr}\,\tilde{\sigma}_{0}\right) \lim_{\lambda \rightarrow
\infty }\frac{f\left( \lambda \right) }{\lambda }  \label{Dmax-formula}
\end{equation}%
where, \ with $\pi _{X}$ denoting the projector onto $\mathrm{supp}\,X$, 
\begin{eqnarray*}
\sigma _{0,11} &:&=\pi _{\sigma _{1}}\sigma _{0}\pi _{\sigma
_{1}},\,\,\sigma _{0,12}:=\pi _{\sigma _{1}}\sigma _{0}\left( \mathbf{1}-\pi
_{\sigma _{1}}\right) , \\
\,\,\sigma _{0,21} &:&=\left( \mathbf{1}-\pi _{\sigma _{1}}\right) \sigma
_{0}\pi _{\sigma _{1}},\,\sigma _{0,22}:=\left( \mathbf{1}-\pi _{\sigma
_{1}}\right) \sigma _{0}\left( \mathbf{1}-\pi _{\sigma _{1}}\right) , \\
\tilde{\sigma}_{0} &=&\sigma _{0,11}-\sigma _{0,12}\left( \sigma
_{0,22}\right) ^{-1}\sigma _{0,21}.
\end{eqnarray*}%
Observe, if $\sigma _{0}$ is invertible, 
\[
\tilde{\sigma}_{0}=\left( \pi _{\sigma _{1}}\sigma _{0}^{-1}\pi _{\sigma
_{1}}\right) ^{-1}. 
\]%
The following property of $D_{f}^{\max }$ will turn out to be useful.

\begin{lemma}
\label{lem:Dmax}\cite{Matsumoto:max-f}(i) If a real valued two-point
function $D_{f}^{Q}\left( \cdot ||\cdot \right) $ of operators is monotone
decreasing by application of CPTP maps,%
\[
D_{f}^{Q}\left( \Lambda \left( \sigma _{0}\right) ||\Lambda \left( \sigma
_{1}\right) \right) \leq D_{f}^{Q}\left( \sigma _{0}||\sigma _{1}\right) 
\]%
and coincide with $D_{f}\left( \cdot ||\cdot \right) $ on commutative
subalgebra, then 
\[
D_{f}^{Q}\left( \sigma _{0}||\sigma _{1}\right) \leq D_{f}^{\max }\left(
\sigma _{0}||\sigma _{1}\right) . 
\]%
(ii) There is a pair $\left( \left\{ q_{\theta }\right\} _{\theta \in \Theta
},\Gamma \right) $ which satisfies 
\[
D_{f}\left( q_{0}||q_{1}\right) =\mathrm{tr}\,\sigma _{1}f\left( \sigma
_{1}^{-1/2}\tilde{\sigma}_{0}\sigma _{1}^{-1/2}\right) +\left( 1-\mathrm{tr}%
\,\tilde{\sigma}_{0}\right) \lim_{\lambda \rightarrow \infty }\frac{f\left(
\lambda \right) }{\lambda } 
\]%
for all convex functions on $[0,\infty )$ at the same time.
\end{lemma}

\bigskip

In the paper we sometimes use the following family of operator convex
functions 
\[
f^{s}\left( \lambda \right) :=-\lambda ^{s},\,\left( 0<s\leq 1\right) . 
\]%
If $s<1$, $\lim_{\lambda \rightarrow \infty }\frac{f\left( \lambda \right) }{%
\lambda }=0$. Thus, 
\begin{eqnarray}
\lim_{s\uparrow 1}D_{f^{s}}\left( p_{0}||p_{1}\right) &=&-\lim_{s\uparrow
1}\sum_{x\in \mathrm{supp\,}p_{1}}\,\left( p_{1}\left( x\right) \right)
^{1-s}\left( p_{0}\left( x\right) \right) ^{s}  \nonumber \\
&=&-\sum_{x\in \mathrm{supp\,}p_{1}}\,p_{0}\left( x\right) ,\,\,\,  \nonumber
\\
\lim_{s\uparrow 1}D_{f^{s}}^{\max }\left( \sigma _{0}||\sigma _{1}\right)
&=&-\lim_{s\uparrow 1}\mathrm{tr}\,\sigma _{1}\left( \sigma _{1}^{-1/2}%
\tilde{\sigma}_{0}\sigma _{1}^{-1/2}\right) ^{s}  \nonumber \\
&=&-\mathrm{tr}\,\tilde{\sigma}_{0}.  \label{Dfs}
\end{eqnarray}%
Meantime, 
\[
D_{f^{1}}\left( p_{0}||p_{1}\right) =D_{f^{1}}^{\max }\left( \sigma
_{0}||\sigma _{1}\right) =-1. 
\]

\section{A necessary and sufficient condition}

To obtain the necessary and sufficient condition for the existence of a CPTP
map $\Gamma $ with (\ref{map}), we use the quantum randomization criterion%
\cite{Jencova:14}\cite{Matsumoto:random}. Let $\mathcal{H}_{D}$ be a Hilbert
space, and $\mathfrak{W}=\left\{ W_{\theta }\right\} _{\theta \in \Theta }$
be a pair of (bounded) operators on $\mathcal{H}_{D}$. We define 
\[
D_{\mathfrak{W}}\left( \sigma _{0}||\sigma _{1}\right) :=\max \left\{
\sum_{\theta \in \Theta }\mathrm{tr}\,W_{\theta }\Lambda \left( \sigma
_{\theta }\right) \,;\Lambda \text{: CPTP from }\mathcal{L}\left( \mathcal{H}%
\right) \text{ to }\mathcal{L}\left( \mathcal{H}_{D}\right) \right\} . 
\]%
Then a CPTP map $\Gamma $ with (\ref{map}) exists if and only if 
\[
D_{\mathfrak{W}}\left( p_{0}||p_{1}\right) \geq D_{\mathfrak{W}}\left(
\sigma _{0}||\sigma _{1}\right) 
\]%
holds for all $\mathfrak{W}$ and all (finite dimensional) $\mathcal{H}_{D}$.

Observe%
\begin{eqnarray*}
D_{\mathfrak{W}}\left( p_{0}||p_{1}\right) &=&\max_{\Lambda \text{:CPTP}%
}\sum_{\theta \in \Theta }\sum_{x\in \mathcal{X}}p_{\theta }\left( x\right) 
\mathrm{tr}\,W_{\theta }\Lambda \left( \left\vert e_{x}\right\rangle
\left\langle e_{x}\right\vert \right) , \\
&=&\max_{\left\{ \rho _{x}\right\} \text{: states on }\mathcal{H}_{D}\text{ }%
}\sum_{\theta \in \Theta }\sum_{x}p_{\theta }\left( x\right) \mathrm{tr}%
\,W_{\theta }\rho _{x}, \\
&=&\sum_{x\in \mathcal{X}}\max_{\rho \text{: a state on }\mathcal{H}_{D}%
\text{ }}\mathrm{tr}\sum_{\theta \in \Theta }p_{\theta }\left( x\right)
\,W_{\theta }\rho , \\
&=&\sum_{x\in \mathcal{X}}r_{\max }\left( \sum_{\theta \in \Theta }p_{\theta
}\left( x\right) \,W_{\theta }\right) \\
&=&\sum_{x\in \mathrm{supp\,}p_{1}}\,p_{1}\left( x\right) f\left( \frac{%
p_{0}\left( x\right) }{p_{1}\left( x\right) }\right) +\left( \sum_{x\not\in 
\mathrm{supp\,}p_{1}}\,p_{0}\left( x\right) \right) \max_{\rho \text{: state
on }\mathcal{H}_{D}\text{ }}\mathrm{tr}\,W_{0}\rho .
\end{eqnarray*}%
where $r_{\max }\left( X\right) $ is the largest eigenvalue of $X$, and 
\[
f\left( \lambda \right) :=r_{\max }\left( \lambda W_{0}+\,W_{1}\right) . 
\]%
Since 
\[
\lim_{\lambda \rightarrow \infty }\frac{f\left( \lambda \right) }{\lambda }%
=\lim_{\lambda \rightarrow \infty }r_{\max }\left( W_{0}+\,\frac{1}{\lambda }%
W_{1}\right) =r_{\max }\left( \,W_{0}\right) , 
\]%
we have%
\[
D_{\mathfrak{W}}\left( p_{0}||p_{1}\right) =D_{f}\left( p_{0}||p_{1}\right)
. 
\]

Note, there are many $D_{\mathfrak{W}}\left( \cdot ||\cdot \right) $ whose
restriction equals $D_{f}\left( \cdot ||\cdot \right) $. Since $D_{\mathfrak{%
W}}\left( \cdot ||\cdot \right) $ is monotone decreasing by application of
CPTP maps almost by definition, they are all bounded from above by $%
D_{f}^{\max }\left( \cdot ||\cdot \right) $ : 
\[
D_{f}^{\max }\left( \sigma _{0}||\sigma _{1}\right) \geq D_{\mathfrak{W}%
}\left( \sigma _{0}||\sigma _{1}\right) . 
\]%
holds. Therefore, if 
\begin{equation}
D_{f}\left( p_{0}||p_{1}\right) \geq D_{f}^{\max }\left( \sigma _{0}||\sigma
_{1}\right)  \label{c-q-monotone}
\end{equation}%
holds for any closed proper convex function $f$ on $[0,\infty )$, a CPTP map
with (\ref{map}) exists. Since $D_{f}^{\max }\left( \cdot ||\cdot \right) $
is monotone decreasing by application of CPTP maps, this condition is
obviously necessary. Thus:

\begin{theorem}
\label{th:all-convex}A CPTP map $\Gamma $ with (\ref{map}) exists if and
only if (\ref{c-q-monotone}) \ holds for any closed proper convex function $%
f $ on $[0,\infty )$.
\end{theorem}

To our regret, no closed formula of $D_{f}^{\max }\left( \cdot ||\cdot
\right) $ had been found out unless $f$ is operator convex. Indeed, we have
the following negative implication (The proof is done later):

\begin{proposition}
\label{prop:convex-necessary}If \ (\ref{Dmax-formula}) holds for any
positive operators $\rho $ and $\sigma $ which not necessarily with unit
trace, then $f$ has to be operator convex.
\end{proposition}

\section{Sufficient conditions}

Lemma\thinspace \ref{lem:Dmax}, (ii) implies an upper bound to $D_{f}^{\max
}\left( \sigma _{0}||\sigma _{1}\right) $. Therefore, we have the following
sufficient condition.

\begin{corollary}
If 
\[
D_{f}\left( p_{0}||p_{1}\right) \leq \mathrm{tr}\,\sigma _{1}f\left( \sigma
_{1}^{-1/2}\tilde{\sigma}_{0}\sigma _{1}^{-1/2}\right) +\left( 1-\mathrm{tr}%
\,\tilde{\sigma}_{0}\right) \lim_{\lambda \rightarrow \infty }\frac{f\left(
\lambda \right) }{\lambda } 
\]%
holds for any closed proper convex function $f$ on $[0,\infty )$, (\ref{map}%
) holds.
\end{corollary}

There is a sufficient condition which can be described only using operator
convex functions, where the formula (\ref{Dmax-formula}) applies.

\begin{lemma}
\label{lem:f-finite}(Lemma\thinspace 5.2 of \cite{HiaiMosonyiPetzBeny}) If $%
f $ is a complex valued function on finitely many points $\left\{ \lambda
_{i};i\in I\right\} $ $\subset \lbrack 0,\infty )$, then for any pairwise
different positive numbers $\left\{ t_{i};i\in I\right\} $ there exist
complex numbers $\left\{ c_{i};i\in I\right\} $ such that $f\left( \lambda
_{i}\right) =\sum_{j\in I}\frac{c_{j}}{\lambda _{i}+t_{j}}$ $i\in I$.
\end{lemma}

\begin{theorem}
If 
\begin{equation}
D_{f}\left( p_{0}||p_{1}\right) =D_{f}^{\max }\left( \sigma _{0}||\sigma
_{1}\right)  \label{D=D}
\end{equation}%
for any operator convex function $f$ on $[0,\infty )$, a CPTP map $\Gamma $
with (\ref{map}) exists. In fact, one only has to check identity for $\left(
t+\lambda \right) ^{-1}$ and $-\lambda ^{s}\,$, where $t\geq 0$ and $s\in
\left( s_{0},1\right) $. Here $s_{0}$ is an arbitrary positive number
smaller than $1$.
\end{theorem}

\begin{proof}
Let 
\[
\left\{ \lambda _{i};i\in I\right\} :=\left\{ \frac{p_{0}\left( x\right) }{%
p_{1}\left( x\right) };x\in \mathrm{supp\,}p_{1}\,\right\} \cup \mathrm{spec}%
\left\{ \sigma _{1}^{-1/2}\tilde{\sigma}_{0}\sigma _{1}^{-1/2}\right\} . 
\]%
and apply Lemma\thinspace \ref{lem:f-finite}. Suppose (\ref{D=D}) holds for $%
\left( t+\lambda \right) ^{-1}$ , $\forall t\geq 0$. Then 
\[
\sum_{x\in \mathrm{supp\,}p_{1}}\,p_{1}\left( x\right) f\left( \frac{%
p_{0}\left( x\right) }{p_{1}\left( x\right) }\right) =\mathrm{tr}\,\sigma
_{1}f\left( \sigma _{1}^{-1/2}\tilde{\sigma}_{0}\sigma _{1}^{-1/2}\right) 
\]%
holds for any convex function $f$ on $[0,\infty )$. Suppose (\ref{D=D})
holds for $f^{s}\left( \sigma \right) =-\lambda ^{s}$ , $\forall s\in \left(
s_{0},1\right) $. Then considering $s\uparrow 1$, by (\ref{Dfs}),%
\[
\sum_{x\in \mathrm{supp\,}p_{1}}\,p_{0}\left( x\right) =\mathrm{tr}\,\tilde{%
\sigma}_{0}. 
\]%
Therefore, for any $f$, 
\[
\left( \sum_{x\not\in \mathrm{supp\,}p_{1}}\,p_{0}\left( x\right) \right)
\lim_{\lambda \rightarrow \infty }\frac{f\left( \lambda \right) }{\lambda }%
=\left( 1-\mathrm{tr}\,\tilde{\sigma}_{0}\right) \lim_{\lambda \rightarrow
\infty }\frac{f\left( \lambda \right) }{\lambda }. 
\]%
Summing up, we have (\ref{D=D}) for any convex function $f$ on $[0,\infty )$%
. Then application of Theorem\thinspace \ref{th:all-convex} leads to the
assertion.
\end{proof}

\section{A necessary and sufficient condition for special case}

There is a case where we can give "tractable" necessary and sufficient
condition. An example is the case where $\sigma _{1}$ is a pure state (the
dimension of the Hilbert space $\mathcal{H}$ is arbitrary finite integer)
Then 
\[
\tilde{\sigma}_{0}=\sigma _{1}^{-1/2}\tilde{\sigma}_{0}\sigma
_{1}^{-1/2}=\gamma \sigma _{1}, 
\]%
where 
\[
\gamma :=\sigma _{0,11}-\sigma _{0,12}\left( \sigma _{0,22}\right)
^{-1}\sigma _{0,21}. 
\]%
Therefore, by (\ref{Dmax-formula}),%
\[
D_{f}^{\max }\left( \sigma _{0}||\sigma _{1}\right) =f\left( \gamma \right)
+\left( 1-\gamma \right) \lim_{\lambda \rightarrow \infty }\frac{f\left(
\lambda \right) }{\lambda }. 
\]%
Suppose (\ref{map}) holds. With $\delta _{x_{0}}$ being a delta distribution
concentrated on $x_{0}$, we should have 
\[
\Gamma \left( \delta _{x}\right) =\sigma _{1},\,\,\forall x\in \mathrm{supp}%
\,p_{1}, 
\]%
since $\sigma _{1}$ is rank\thinspace -\thinspace 1 projector. Therefore, 
\[
\sigma _{0}=\Gamma \left( p_{0}\right) =\sum_{x\in \mathrm{supp}%
\,p_{1}}p_{0}\left( x\right) \sigma _{1}+\sum_{x\notin \mathrm{supp}%
\,p_{1}}p_{0}\left( x\right) \Gamma \left( \delta _{x}\right) . 
\]%
For this to hold for some choice of $\Gamma \left( \delta _{x}\right) $ ($%
x\notin \mathrm{supp}\,p_{1}$), it is necessary and sufficient that 
\[
\sigma _{0}-\sum_{x\in \mathrm{supp}\,p_{1}}p_{0}\left( x\right) \sigma
_{1}\geq 0 
\]%
holds. (Necessity is trivial. On the other hand, if this inequality holds,
we only have to define 
\[
\Gamma \left( \delta _{x}\right) :=\frac{1}{\sum_{x\notin \mathrm{supp}%
\,p_{1}}p_{0}\left( x\right) }\left( \sigma _{0}-\sum_{x\in \mathrm{supp}%
\,p_{1}}p_{0}\left( x\right) \sigma _{1}\right) ,\,\,x\notin \mathrm{supp}%
\,p_{1}\text{.} 
\]%
)

A necessary and sufficient condition of this is 
\begin{equation}
\sum_{x\in \mathrm{supp}\,p_{1}}p_{0}\left( x\right) \leq \gamma .
\label{sum-p}
\end{equation}%
On the other hand, consider $f^{s}\left( \lambda \right) :=-\lambda
^{s}\,\left( 0<s<1\right) $, which is operator convex. Suppose 
\[
\sum_{x\in \mathrm{supp}\,p_{1}}p_{1}\left( x\right) f^{s}\left( \frac{%
p_{0}\left( x\right) }{p_{1}\left( x\right) }\right) \geq f^{s}\left( \gamma
\right) 
\]%
holds for all $0<s<1$, Then \ letting $s\uparrow 1$, we have 
\[
-\sum_{x\in \mathrm{supp}\,p_{1}}p_{0}\left( x\right) \geq -\left( \gamma
\right) , 
\]%
which is (\ref{sum-p}), or equivalently, (\ref{map}). The result above is
summarized as follows.

\begin{proposition}
When $\sigma _{1}$ is a pure state, then a CPTP map $\Gamma $ with (\ref{map}%
) exists if and only if (\ref{c-q-monotone}) with $f\left( \lambda \right)
=-\lambda ^{s}$ for all $s\in \left( s_{0},1\right) $. Here $s_{0}$ is an
arbitrary positive number smaller than $1$.
\end{proposition}

\section{Operator convex functions are not enough}

A bad news is that (\ref{c-q-monotone}) for all operator convex functions is
not enough to show the existence of a CPTP map $\Gamma $ with (\ref{map}). A
counter example is constructed by letting both $\left\{ p_{\theta }\right\}
_{\theta \in \Theta }$ and $\left\{ \sigma _{\theta }\right\} _{\theta \in
\Theta }$ be probability distributions on 3-\thinspace points set $\left\{
1,2,3\right\} $. In addition we suppose $p_{1}$ and $\sigma _{1}$ are
uniform distributions, and parameterize $p_{0}$ and $\sigma _{0}$ by 
\[
\sigma _{0}=\left( a,b,c\right) ,\,p_{0}:=\left(
a_{0},\,b_{0}\,,c_{0}\right) ,\, 
\]%
where $c=1-a-b$, \ $c_{0}=1-a_{0}-b_{0}$ and 
\[
a_{0}<b_{0}<c_{0}\,. 
\]

Since the uniform distribution is a fixed point, the stochastic map which
sends $p_{\theta }$ to $q_{\theta }$ is doubly stochastic, or equivalently,
a convex combination of permutations. Therefore, $\sigma _{0}$ has to be in
the convex hull of six points, $\left( a_{0}\,,b_{0},\,c_{0}\right) $, $%
\left( a_{0},\,c_{0},\text{ }b_{0}\,\right) $, $\left(
b_{0}\,,a_{0}\,,c_{0}\right) $, and so on.

\begin{lemma}
\bigskip \label{lem:lowner-2}(Theorem\thinspace 8.1 if \cite%
{HiaiMosonyiPetzBeny}) A continuous real valued function $f$ on $[0,\infty )$
is operator convex if and only if 
\[
f\left( \lambda \right) =f\left( 0\right) +\alpha \lambda +\beta \lambda
^{2}+\int_{\left( 0,\infty \right) }\left( \frac{\lambda }{1+t}-\frac{%
\lambda }{\lambda +t}\right) \mathrm{d}\mu \left( t\right) , 
\]%
where $\alpha $ is a real number, $\beta $ is a non-negative real number,
and $\mu $ is a finite non-negative measure satisfying 
\[
\int_{\left( 0,\infty \right) }\frac{\mathrm{d}\mu \left( t\right) }{\left(
1+t\right) ^{2}}<\infty . 
\]
\end{lemma}

\bigskip

By Lemma\thinspace \ref{lem:lowner-2}, instead of all operator convex
functions, we only have to check (\ref{c-q-monotone}) for $\lambda ^{2}$, $%
\frac{1}{\lambda +t}$. Let 
\begin{eqnarray*}
g_{t}\left( a,b\right) &:&=\frac{1}{a+t}+\frac{1}{b+t}+\frac{\text{ }1}{%
1-a-b+t}-\left\{ \frac{1}{a_{0}+t}+\frac{1}{b_{0}+t}+\frac{1}{c_{0}+t}%
\right\} ,\,\,\left( t\geq 0\right) \\
g_{-1}\left( a,b\right) &:&=a^{2}+b^{2}+\left( 1-a-b\right)
^{2}-a_{0}^{2}-b_{0}^{2}-c_{0}^{2}.
\end{eqnarray*}%
So our purpose is to prove that the set 
\[
\mathcal{C}_{2}:=\bigcap_{t:t\geq 0,t=-1}\{\left( a,b\right)
\,;\,g_{t}\left( a,b\right) \leq 0\}, 
\]%
is not identical to the projection $\mathcal{C}_{1}$ of the convex hull of
the six points to $\left( a,b\right) $-plain. Note the set $\mathcal{C}_{2}$
is convex, and contains the six points. Hence, our task is to find a point
of the set $\mathcal{C}_{2}$ which is not in $\mathcal{C}_{1}$. Observe that
the vertices of $\mathcal{C}_{1}$ are 
\[
\left( a_{0},c_{0}\right) ,\left( b_{0},c_{0}\right) ,\left(
c_{0},b_{0}\right) ,\left( c_{0},a_{0}\right) ,\left( b_{0},a_{0}\right)
,\left( a_{0},b_{0}\right) \text{,} 
\]%
the maximum of $b$-coordinate of $\mathcal{C}_{1}$ is $c_{0}$, and the edge
connecting $\left( a_{0},c_{0}\right) $ and $\left( b_{0},c_{0}\right) $
forms the "upper bound" of $\mathcal{C}_{1}$. Hence, we only have to show
that there is a point in $\mathcal{C}_{2}$ whose $b$-coordinate is strictly
larger than $c_{0}$.

Observe also the line%
\begin{equation}
b=1-2a  \label{b=1-2a}
\end{equation}%
intersects with the edge connecting $\left( a_{0},c_{0}\right) $ and $\left(
b_{0},c_{0}\right) $ at $\left( \frac{a_{0}+b_{0}}{2},c_{0}\right) $. Thus
this line intersects $g_{t}\left( a,b\right) =0$ in the region above $%
b=c_{0} $. Denote the intersection point $\left( a_{t},b_{t}\right) $. Then
since the line segment connecting $\left( a_{t},b_{t}\right) $ and $\left( 
\frac{a_{0}+b_{0}}{2},c_{0}\right) $ is in the set $\{\left( a,b\right)
\,;\,g_{t}\left( a,b\right) \leq 0\}$ , $\mathcal{C}_{2}$ contains the line
segment connecting $\left( a_{\ast },b_{\ast }\right) $, where 
\[
b_{\ast }:=\inf_{t:t\geq 0,t=-1}b_{t}. 
\]%
So we only have to show $b_{\ast }>c_{0}$.

Solving $g_{t}\left( \frac{1}{2}\left( 1-b_{t}\right) ,b_{t}\right) =0$,%
\[
b_{t}=\frac{1}{2\left( 3t^{2}-t+e_{t}\right) }\left\{ 
\begin{array}{c}
2t^{2}\allowbreak +\left( e_{t}-1\right) t+e_{t} \\ 
\pm \sqrt{\left( 24e_{t}-8\right) t^{4}+8e_{t}t^{3}+\left(
9e_{t}^{2}-6e_{t}+1\right) \allowbreak t^{2}+\left( 6e_{t}^{2}-2e_{t}\right)
t+e_{t}^{2}}%
\end{array}%
\right\} , 
\]%
where $e_{t}$ is defined by the identity 
\[
\frac{1}{t}\left( 3-\frac{1}{t}+\frac{e_{t}}{t^{2}}\right) =\frac{1}{a_{0}+t}%
+\frac{1}{b_{0}+t}+\frac{1}{c_{0}+t}. 
\]%
Here, let $t\rightarrow \infty $. Then $e_{t}\rightarrow
a_{0}^{2}+b_{0}^{2}+c_{0}^{2}$ and 
\[
b_{t}\rightarrow \frac{1}{3}\left( 1\pm \sqrt{6\left(
a_{0}^{2}+b_{0}^{2}+c_{0}^{2}\right) -2}\right) . 
\]%
Elementally calculations show that the larger solution of the two is
strictly larger than $c_{0}$. Thus, there is $t_{0}>0$ such that 
\[
\inf_{t>t_{0}}\,b_{t}>\frac{1}{2}\left( c_{0}+\lim_{t\rightarrow \infty
}b_{t}\right) . 
\]%
Therefore, 
\[
b_{\ast }>\min \left\{ \frac{1}{2}\left( c_{0}+\lim_{t\rightarrow \infty
}b_{t}\right) ,\inf_{t\in \lbrack 0,t_{0}]\cup \left\{ -1\right\}
}\,b_{t}\right\} . 
\]%
Since the function $t\rightarrow b_{t}$ is continuous, there is $t_{\ast }$
such that%
\[
\inf_{t\in \lbrack 0,t_{0}]\cup \left\{ -1\right\} }\,b_{t}=b_{t_{\ast }}\,. 
\]

Since $g_{t_{\ast }}\left( a,b\right) =0$ is an algebraic convex curve and
passes through the two points $\left( a_{0},c_{0}\right) $ and $\left(
b_{0},c_{0}\right) $, it cannot coincide with the line connecting these two
points. Thus, $b_{t_{\ast }}>c_{0}$. Therefore, we have $b_{\ast }>c_{0}$,
and $\mathcal{C}_{1}\neq \mathcal{C}_{2}$. Thus (\ref{c-q-monotone}) for all
operator convex functions is not enough for the existence of {}a CPTP map $%
\Gamma $ with (\ref{map}).

\section{Proof of Proposition\thinspace \protect\ref{prop:convex-necessary}}

Suppose $\sigma _{1}>0$ and a CPTP map $\Gamma $ satisfies (\ref{map}). Also
let

\begin{eqnarray*}
M_{x} &:&=p_{1}\left( x\right) \sigma _{1}^{-1/2}\Gamma \left( \left\vert
e_{x}\right\rangle \left\langle e_{x}\right\vert \right) \sigma _{1}^{-1/2},
\\
d &:&=\sigma _{1}^{-1/2}\sigma _{0}\sigma _{1}^{-1/2}, \\
r\left( x\right) &:&=\frac{p_{0}\left( x\right) }{p_{1}\left( x\right) }.
\end{eqnarray*}%
Then 
\begin{eqnarray*}
\sum_{x\in \mathcal{X}}M_{x} &=&\sum_{x}p_{1}\left( x\right) \sigma
_{1}^{-1/2}\Gamma \left( \left\vert e_{x}\right\rangle \left\langle
e_{x}\right\vert \right) \sigma _{1}^{-1/2} \\
&=&\sigma _{1}^{-1/2}\sigma _{1}\sigma _{1}^{-1/2}=\pi _{\sigma _{1}}, \\
\sum_{x\in \mathcal{X}}r\left( x\right) M_{x} &=&d, \\
\mathrm{tr}\,\sigma _{1}M_{x} &=&p_{1}\left( x\right) .
\end{eqnarray*}%
Also, 
\begin{eqnarray*}
D_{f}\left( p_{0}||p_{1}\right) &=&\sum_{x\in \mathcal{X}}p_{1}\left(
x\right) f\left( r\left( x\right) \right) \\
&=&\mathrm{tr}\,\sigma _{1}\sum_{x\in \mathcal{X}}f\left( r\left( x\right)
\right) M_{x},
\end{eqnarray*}%
and 
\[
\mathrm{tr}\,\sigma _{1}f\left( d\right) =\mathrm{tr}\,\sigma _{1}f\left(
\sum_{x\in \mathcal{X}}r\left( x\right) M_{x}\right) . 
\]%
Therefore, 
\begin{eqnarray*}
&&\min \left\{ D_{f}\left( p_{0}||p_{1}\right) ;\Gamma ,\left\{ p_{\theta
}\right\} _{\theta \in \Theta }\text{ with (\ref{map})}\right\} \\
&=&\min \left\{ \mathrm{tr}\,\sigma _{1}\sum_{x\in \mathcal{X}}f\left(
r\left( x\right) \right) M_{x};\,\sum_{x\in \mathcal{X}}M_{x}=\pi
_{x},\,d=\sum_{x\in \mathcal{X}}r\left( x\right) M_{x}\right\} \\
&=&\min \left\{ \mathrm{tr}\,\sigma _{1}V^{\dagger }f\left( d^{\prime
}\right) V;\,V\text{: isometry from }\mathcal{H}^{\prime }\text{ to }%
\mathcal{H}\text{, }d=V^{\dagger }d^{\prime }V\right\}
\end{eqnarray*}%
The second identity above is by Naimark extension theorem, which states
there is a Hilbert space $\mathcal{H}^{\prime }$ which is larger in
dimension than $\mathcal{H}$ , and the projectors $\left\{ P_{x}\right\} $ ,
the isometry $V$ from $\mathcal{H}^{\prime }$ to $\mathcal{H}$ such that%
\[
M_{x}=V^{\dagger }P_{x}V. 
\]%
Note that $\mathcal{H}^{\prime }$, $d^{\prime }$, $V$ is not restricted only
by $d$, and not by $\sigma _{1}$. (Here, recall we had removed the
restriction of trace of $\sigma _{0}$, so that $d$ and $\sigma _{1}$ can
move freely.)

For fixed $\left( \mathcal{H}^{\prime },d^{\prime },V\right) $, if the
inequality 
\[
\mathrm{tr}\,\sigma _{1}f\left( d\right) =\mathrm{tr}\,\sigma _{1}f\left(
V^{\dagger }d^{\prime }V\right) \leq \mathrm{tr}\,\sigma _{1}V^{\dagger
}f\left( d^{\prime }\right) V 
\]%
holds true for any $\sigma _{1}>0$, we should have 
\[
f\left( V^{\dagger }d^{\prime }V\right) \leq V^{\dagger }f\left( d^{\prime
}\right) V. 
\]%
If this holds for any $\left( \mathcal{H}^{\prime },d^{\prime },V\right) $,
then $f$ has to be operator convex (See Exercise V.2.2 of \cite{Bhatia}, for
example).

\end{document}